\documentclass[aps,prl,twocolumn,superscriptaddress]{revtex4}

\usepackage{amsfonts,amsmath,amssymb,amsthm,bm,bbm,graphicx,color}

\newcommand{\sBQP}{{\textsf{\#BQP}}}
\newcommand{\QMA}{\textsf{QMA}}
\newcommand{\sP}{\textsf{\#P}}
\newcommand{\PP}{\textsf{PP}}
\newcommand{\NP}{\textsf{NP}}

\newcommand{\ket}[1]{\lvert#1\rangle}
\newcommand{\bra}[1]{\langle#1\rvert}
\newcommand{\poly}{\mathrm{poly}}
\newcommand{\mc}[1]{\mathcal{#1}}
\newcommand{\DOS}{\textsc{dos}}
\newcommand{\sLH}{\textsc{\#lh}}
\newcommand{\ketanc}{\ket{\bm{0}}_A}
\newcommand{\braanc}{\bra{\bm 0}_A}
\newcommand{\tr}{\mathrm{tr\:}}
\newcommand{\ketbra}[2]{\lvert{#1}\rangle\!\langle{#2}\rvert}

\newcommand{\proj}[1]{\ketbra{#1}{#1}}
\newcommand{\abs}[1]{\lvert#1\rvert}

\newcommand{\one}{\mathbbm{1}}
\renewcommand{\openone}{\one}

 \newtheorem{theorem}{Theorem}
 \newtheorem{definition}[theorem]{Definition}
 \newtheorem{lemma}[theorem]{Lemma}
 \renewenvironment{proof}[1][Proof]{\noindent\textbf{#1.} }{\,$\Box$\\}

\begin{document}

\title{Computational Difficulty of Computing the Density of States}

\author{Brielin Brown}
\affiliation{University of Virginia, Departments of Physics and Computer Science, 
Charlottesville, Virginia 22904, USA}
\affiliation{Perimeter Institute for Theoretical Physics, 
Waterloo, Ontario N2L 2Y5, Canada}
\author{Steven T.\ Flammia}
\affiliation{Perimeter Institute for Theoretical Physics, 
Waterloo, Ontario N2L 2Y5, Canada}
\author{Norbert Schuch}
\affiliation{California Institute of Technology,
    Institute for Quantum Information,
    MC 305-16, Pasadena, California 91125, USA}

\begin{abstract}
We study the computational difficulty of computing the ground state
degeneracy and the density of states for local Hamiltonians. We show that
the difficulty of both problems is exactly captured by a class which we
call \sBQP, which is the counting version of the quantum complexity class
\QMA. We show that \sBQP{} is not harder than its classical counting
counterpart \sP, which in turn implies that computing the ground state
degeneracy or the density of states for classical Hamiltonians is just as
hard as it is for quantum Hamiltonians.
\end{abstract}

\maketitle

Understanding the physical properties of correlated quantum many-body
systems is a problem of central importance in condensed matter physics.
The density of states, defined as the number of energy eigenstates per
energy interval, plays a particularly crucial role in this endeavor.  It
is a key ingredient when deriving many thermodynamic properties from
microscopic models, including specific heat capacity, thermal
conductivity, band structure, and (near the Fermi energy) most electronic
properties of metals. Computing the density of states can be a daunting
task however, as it in principle involves diagonalizing a Hamiltonian
acting on an exponentially large space, though other more efficient
approaches which might take advantage of the structure of a given problem
are not a priori ruled out.

In this Letter, we precisely quantify the difficulty of computing the
density of states by using the powerful tools of quantum complexity
theory. Quantum complexity aims at generalizing the well-established field
of classical complexity theory to assess the difficulty of tasks related
to quantum mechanical problems, concerning both the classical difficulty
of simulating quantum systems as well as the fundamental limits to the
power of quantum computers.  In particular, quantum complexity theory has
managed to explain the difficulty of computing ground state
properties of quantum spin systems in various settings, such as
two-dimensional (2D) lattices~\cite{terhal:lh-2d-qma} and even
one-dimensional (1D) chains~\cite{aharonov:1d-qma}, as well as fermionic
systems~\cite{schuch:dft-qma}.

We will determine the computational difficulty of two problems: First,
computing the density of states of a local Hamiltonian, and second,
counting the ground state degeneracy of a local gapped Hamiltonian; 
in both cases, the result holds even if the Hamiltonian is restricted to
act on a 2D lattice of qubits, or on a 1D chain.
To this end, we will introduce the quantum counting class \sBQP{} (sharp
BQP), which constitutes the natural counting version of the class \QMA{}
(Quantum Merlin Arthur) which itself captures the difficulty of computing
the ground state energy of a local
Hamiltonian~\cite{kitaev:book,kitaev:qma}. Vaguely speaking, \sBQP{}
counts the number of possible ``quantum solutions'' to a quantum problem
that can be verified using a quantum computer.  We show that both
problems, computing the density of states and counting the ground state
degeneracy, are complete problems for the class \sBQP, i.e., they are
among the hardest problems in this class. 

Having quantified the difficulty of computing the density of states and
counting the number of ground states, we proceed to relate \sBQP{} to
known classical counting complexity classes, and show
that \sBQP{} equals \sP{} (under weakly parsimonious reductions).
Here, the complexity class \sP\ counts the number of satisfying
assignments
to any efficiently computable boolean function. This can be a very hard
problem which is believed to take exponential time; in particular, it is
at least as hard as deciding whether the function has at least one
satisfying input, i.e., the complexity class
\NP. Examples for \sP-complete problems (i.e., the hardest problems in
that class) include counting the number of colorings of a graph, or
computing the permanent of a matrix with binary entries.
  Phrased in terms of Hamiltonians, what we show
is that computing the density of states and counting the ground state
degeneracy of a classical spin system is just as hard as
solving the same problem for a quantum Hamiltonian.

\emph{Quantum complexity classes.---}%
Let us start by introducing the relevant complexity classes. The central
role in the following is taken by the \emph{verifier} $V$, which verifies
``quantum solutions'' (also called \emph{proofs}) to a given problem. More
formally, a verifier checking an $n$-qubit quantum proof (that is, a
quantum state $\ket{\psi}$) consists of a $T=\poly(n)$ length quantum
circuit $U=U_T\cdots U_1$ (with local gates $U_t$) acting on $m=\poly(n)$
qubits, which takes the $n$-qubit quantum state $\ket{\psi}_I$ as an
input, together with $m-n$ initialized ancillas,
$\ketanc\equiv\ket{0\cdots0}_A$, applies $U$, and finally measures the
first qubit in the $\{\ket{0}_1,\ket{1}_1\}$ basis to return $1$ (``proof
accepted'') or $0$ (``proof rejected'').  Then, the class \QMA\ contains
all problems of the form: ``Decide whether there exists a $\ket\psi$ such
that $p_\mathrm{acc}(V(\psi))>a$, or whether $p_\mathrm{acc}(V(\psi))<b$
for all $\ket\psi$, for some chosen $a-b>1/\poly(n)$, given that one is
the case''.  Here, the acceptance probability of a state $\ket\psi$ is
$p_\mathrm{acc}(V(\psi)):= \bra\psi \Omega \ket\psi$, with
\vspace*{-0.05cm}
\begin{equation}
    \label{eq:Omega}
\Omega=(\one_I\otimes\braanc) U^\dagger
(\proj{1}_1\otimes\one)
U(\one_I\otimes\ketanc)\, ,
\end{equation}
\vspace*{-0.05cm}
which we illustrate in Fig.~\ref{F:circuit}.

\begin{figure}[t]
        \includegraphics[width=0.9\columnwidth]{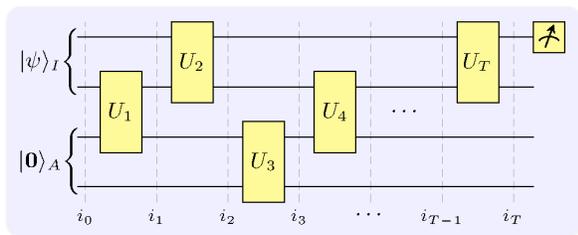}
\caption{\label{F:circuit}%
A \QMA{} verifier consists of a sequence of $T$ local
unitary gates acting on the ``quantum proof'' $\ket{\psi}$ and an
ancillary register initialized to $\ket{\bm{0}}$.  The final measurement
on the first qubit returns $\ket{1}$ or $\ket{0}$ to accept or reject the
proof, respectively. Transition probabilities can be computed by doing a
``path integral'' over all intermediate configurations $(i_k)_k$.
}
\end{figure}

The idea behind this definition is that \QMA{} quantifies the difficulty
of computing the ground state energy $E_0(H)$ of a local Hamiltonian $H$
up to $1/\poly(n)$ accuracy.  Let the verifier be a circuit estimating
$\bra\psi H\ket\psi$; then a black box solving \QMA{} problems can be used
to compute $E_0(H)$ up to $1/\poly(n)$ accuracy by binary search using a
single \QMA{} query. Note also that \QMA{} is the quantum version of the
class \NP, where one is given an efficiently computable boolean function
$f(x)\in\{0,1\}$ and one must figure out if there is an $x$ such that
$f(x)=1$. 

The class \NP{} has a natural counting version, known as \sP.  Here, the
task is to determine the \emph{number} rather than the existence of
satisfying inputs, i.e., to compute 
$\abs{\{x:f(x)=1\}}$.  We will now analogously
define \sBQP, the counting version of \QMA. Consider the verifying map
$\Omega$ of Eq.~\eqref{eq:Omega} for a \QMA{} problem, 
with the additional promise that $\Omega$
does not have eigenvalues between $a$ and $b$, $a-b>1/\poly(n)$.  
Then the class \sBQP{} consists of all problems of the form ``compute
the dimension of the space spanned by all eigenvectors with eigenvalues
$\ge a$''.

An equivalent definition for \sBQP{} (cf.\ also
\cite{brandao:pursuit-for-uniqueness,jain:unique-witness}) is the
following: Consider a verifier $\Omega$ with the additional promise that
there exist subspaces $\mathcal A\oplus \mathcal R = \mathbb C^{2^n}$ such
that $\bra{\psi}\Omega\ket\psi\ge a$ for all $\ket\psi\in \mathcal A$, and
$\bra{\psi}\Omega\ket\psi\le b$ for all $\ket\psi\in \mathcal R$, where
again $a-b>1/\poly(n)$ -- we can think of $\mathcal A$ and $\mathcal R$ as
containing the good and bad witnesses, respectively.  (Note that there
will always be ``mediocre'' witnesses---the question is whether 
there exists a
decomposition into a good and a bad witness space.) Then, \sBQP{}
consists of all problems of the form ``compute $\mathrm{dim}\ \mathcal
A$''.  This number is well-defined, i.e., independent of the choice of
$\mc A$ and $\mc R$, and moreover, one can easily show that it is
equivalent to the definition above, cf.\ the
Supplementary Material.

The gap promise we impose on the spectrum of $\Omega$ is not present in
the definition of \QMA{} (though similarly restricted versions of \QMA{}
were defined in
\cite{brandao:pursuit-for-uniqueness,jain:unique-witness}).  Nevertheless,
this promise emerges naturally when considering the counting version:
\QMA{} captures the difficulty of determining the existence of an input
state with acceptance probability above $a$, up to a ``grace interval'' 
$[b,a]$ in which mistakes are tolerated (i.e., if the largest
eigenvalue of $\Omega$ is in $[b,a]$, the oracle can return either
outcome).  Correspondingly, \sBQP\ captures the difficulty of counting the
number of eigenvalues above $a$, where eivenvalues in the
grace interval $[b,a]$ can be miscounted. The reason why we choose
to define \sBQP\ with a gap promise rather than with a grace interval is
the same as for \QMA, namely to have a unique outcome associated with any
input.

Similarly, the idea of the Hamiltonian formulation of the problem which we
will discuss below is to ask for the number of eigenstates in a certain
energy interval, where states which are in some small $1/\poly(n)$
neighborhood of this interval may be miscounted; again, for reasons of
rigor we choose to consider only Hamiltonians with no eigenstates in that
interval.  It should be noted,  however, that all of the equivalence
proofs we give equally apply if we choose to allow for miscounting of
states in those grace intervals instead of requiring them to be empty,
as the proofs do not make use of the gap promise itself, but rather show that
all states outside those grace intervals are mapped (and thus counted)
correctly.  Thus, while the actual number returned by the grace
interval formulation of the counting problems might change under those
mappings due to different treatment of states in the grace interval, it
will still be in the correct range.

The class \sBQP{} inherits the important property from \QMA{} of being
stable under amplification, that is, the definition of \sBQP{} is not
sensitive to the choice of $a$ and $b$. In particular, any
$a-b>1/\poly(n)$ can be amplified (by building a new poly-size $\Omega'$
from $\Omega$) such that $a'=1-\exp(-\poly(n))$, $b'=\exp(-\poly(n))$, and
keeping the eigenvalue gap between $a'$ and $b'$, by using a construction
called \emph{strong amplification},
cf.~Ref.~\cite{marriott:qma-strongamp}; as shown there, strong
amplification acts on all eigenvalues independently and thus also applies to
\sBQP. The crucial point is that strong amplification works without
changing the proof itself, compared to weak amplification which takes
multiple copies of the proof as an input. While this is fine for \QMA, it does
change the dimension of the accepting subspace in an unpredictable way and
is thus not an option for the amplification of \sBQP.

\emph{Complexity of computing the density of states.---}%
Let us now show why the class \sBQP{} is relevant for physical
applications. In particular, we are going to show that computing the
density of states of a local $n$-spin Hamiltonian $H=\sum_i H_i$ with
few-body terms $H_i$, $\|H_i\|\le 1$, up to accuracy $1/\poly(n)$,  is a problem
which is complete for \sBQP, i.e., it is as hard as any problem in \sBQP{}
can be. The same holds true for the (a priori weaker)
problem of counting the ground state degeneracy of a local Hamiltonian,
given a $1/\poly(n)$ spectral gap above (note that 
Bravyi \emph{et al.}~\cite{Bravyi2009a} suggested this as a 
definition for a quantum counting class).  We can impose
additional restrictions on the interaction structure of our Hamiltonian, and as we
will see, the hardness is preserved even for 2D lattices of
qubits, or 1D systems.

The problem \DOS{} (density of states) is defined as follows: Given a
local Hamiltonian $H=\sum_i H_i$, compute the number of orthogonal
eigenstates with eigenvalues in an
interval $[E_1,E_2]$ with $E_2-E_1>1/\poly(n)$, where we do not allow for
eigenvalues within a small grace interval of width
$\Delta=(E_2-E_1)/\poly(n)$ centerd around $E_1$ and $E_2$; alternatively,
we can allow for errorneous counts of eigenstates in that interval.
Second, the problem \sLH{} (sharp local Hamiltonian) corresponds to
counting the number of ground states of a
local Hamiltonian which has a spectral gap $\Delta=1/\poly(n)$ above the
ground state subspace, given we are told the ground state energy, and
where we allow for a small splitting in the ground state  
energies; again, we can alternatively allow to miscount states in the grace 
interval.

Clearly, \sLH{} is a special instance of \DOS, i.e., solving \sLH{} can be
reduced to solving \DOS. In order to show that \DOS\ is contained in
\sBQP, we can use a phase estimation circuit~\cite{cleve:qalg_revisited}
to estimate the energy of any given input $\ket\psi$ and only accept if
its energy $\bra\psi H\ket\psi$ is in the interval $[E_1,E_2]$; 
as the desired accuracy $\Delta=1/\poly(n)$, this can be done efficiently.
A detailed proof (using a more elementary circuit) is given in the
Supplementary Material.

Let us now conversely show that \sLH{} is a hard problem for \sBQP, that
is, any problem in \sBQP{} can be reduced to counting the ground states of 
some gapped local Hamiltonian~\footnote{
Note that the connection between \QMA{} with a unique ``good witness'',
such as in our \sBQP{} definition, and local Hamiltonians with unique
ground state and a $1/\poly(n)$ gap has been shown
in~\cite{brandao:pursuit-for-uniqueness}.
}. As in turn \sLH{} can be reduced to \DOS,
which is contained in \sBQP,
this proves that both \sLH{} and \DOS{} are complete problems for 
\sBQP, i.e., they capture the full difficulty of
this class. To this end, we need to start from an arbitrary verifier
circuit $U=U_T\cdots U_1$  and construct a
Hamiltonian which has as many ground states as the circuit has accepting
inputs (corresponding to the outcome $\ket{1}_1$ on the first qubit). 
Let $\mc A$ and $\mc R$ be the eigenspaces of $\Omega$
[Eq.~\eqref{eq:Omega}] with eigenvalues $\ge a=1-2^{-\poly(n)}$ and $\le
b=2^{-\poly(n)}$, respectively, and define $U[\mc
R]:=\{U\ket{\psi}_I\ketanc:\ket\psi_I\in \mathcal R\}$.

We will follow Kitaev's original construction for a Hamiltonian
encoding a \QMA{} verifier circuit~\cite{kitaev:book,kitaev:qma}, which for
any proof $\ket\psi_I\in \mc A$ has the ``proof history''
$\ket\Phi=\sum_{t=0}^T U_t\cdots U_1\ket\psi_I\ketanc\ket{t}_T$ as its
ground state, where the third register is used as a clock.  The
Hamiltonian $H=H_\mathrm{init}+\sum_{t=1}^{T}
H_\mathrm{evol}(t)+H_\mathrm{final}$ has three types of
terms: $H_\mathrm{init} = \one\otimes(\one-\proj{\bm
0}_A)\otimes \proj{0}_T$ makes sure the ancilla is initialized, 
$H_\mathrm{evol}(t)=-U_t\otimes \ket{t}\!\bra{t-1}_T+\mathrm{h.c.}$
ensures proper evolution from $t-1$ to $t$, and
$H_\mathrm{final}=\Pi_{U[\mc R]}\otimes \proj{T}_T$ gives an
energy penalty to states $\ket\Phi$ built from proofs
$\ket\psi_I\in\mathcal R$.
Note that our $H_\mathrm{final}$ differs from the
usual choice $\proj{0}_1\otimes \one \otimes \proj{T}_T$ and is in fact
non-local;
as we show in the Supplementary Material, this does not significantly
change the relevant spectral properties (in particular, we keep the
$1/\poly(n)$ gap, and the ground state subspace is split at most
exponentially). With this choice of $H_\mathrm{final}$, $H$ acts
independently on the subspaces spanned by $\{U_t\cdots
U_1\ket{\psi}_I\ket{\bm x}_A\ket{t}_T\}_{t=0,\dots,T}$ for any
$\ket\psi\in \mc A$ or $\ket\psi\in \mc R$, and $\ket{\bm x}_A$
the computational basis, and the restriction of $H$ to any of these
subspaces describes a random walk which is characterized by the choice of
$\ket\psi_I$ and the number of $1$'s in $\ket{\bm x}_A$.
These cases can be analyzed independently (see Supplementary Material),
and it follows that $H$ has a $\dim \mc A$-fold degenerate ground state
space with a $1/\poly(n)$ gap above, proving \sBQP-hardness of \sLH.

This shows that \sLH\ is \sBQP-hard for a Hamiltonian which is a sum of
$\log\, T$-local terms (i.e., each term acts on $\log\, T$ sites),
as the clock register is of size $\log\, T$.  In order to
obtain a $k$-body Hamiltonian, Kitaev suggested to
use a unary encoding of the clock (i.e., $\ket{t}_T$ is encoded as
$\ket{1\cdots10\cdots0}$, with $t$ 1's), so that each Hamiltonian term
only acts on three qubits of the clock.  However, this makes the Hilbert
space of the clock too big, and terms need to be added to the
Hamiltonian to penalize illegal clock configurations.
These terms divide the Hilbert space into two parts, $\mc
H_\mathrm{legal}\oplus \mc
H_\mathrm{illegal}$. Here, $\mc H_\mathrm{legal}$ contains only legal
clock states, whereas $\mc
H_\mathrm{illegal}$ contains only configurations with illegal
clock states~\cite{kitaev:book,kitaev:qma}. Since no Hamiltonian term
couples these two subspaces, the Hamiltonian can be analyzed independently
on the two subspaces. It turns out that its restriction to $\mc
H_\mathrm{illegal}$ has an at least $1/\poly(n)$ higher energy, while
on $\mc H_\mathrm{legal}$,  the Hamiltonian is exactly the
same as before. Thus, one finds that the new Hamiltonian still has the right
number of ground states, and a $1/\poly(n)$ spectral gap.  The very same
argument applies in the case of 1D Hamiltonians, using the
\QMA-construction of Ref.~\cite{aharonov:1d-qma}: Again, the Hamiltonian
acts independently on a ``legal'' and an ``illegal'' subspace, where the
latter has a polynomially larger energy, and the former reproduces the
(low-energy) spectrum of the original Hamiltonian~\cite{schuch:mps-gap-np}.  

An alternative way to prove \QMA-hardness on restricted interaction graphs
is to use so-called \emph{perturbation gadgets}, which yield the
Hamiltonian of the Kitaev construction above from a perturbative
expansion; in
particular, this way one can show \QMA-hardness of Pauli-type
nearest-neighbor Hamiltonians on a 2D square lattice of
qubits~\cite{terhal:lh-2d-qma}. As shown in Ref.~\cite{kempe:3local-qma},
such gadgets do in fact approximately preserve the whole low-energy part
of the spectrum,  and thus, our \sBQP--hardness proof for \sLH\ still
applies to these classes of Hamiltonians. It should be noted, however,
that since the eigenvalues are only preserved up to an error
$1/\poly(n)$, the splitting of the ground state space will now be of order
$1/\poly(n)$; however, it can still be chosen to be polynomially smaller
than the spectral gap.

\emph{Quantum vs.\ classical counting complexity.---}%
As we have seen, the quantum counting class \sBQP{} exactly captures the
difficulty of counting the degeneracy of ground states and computing the
density of states of local quantum Hamiltonians. In the following, we will
relate \sBQP{} to classical counting classes and
prove that \sBQP{} is equal to \sP, counting the number of satisfying
inputs to a boolean function~%
\footnote{%
Formally speaking, the
reduction from \sBQP{} to \sP{} is \emph{weakly parsimonious},
i.e., for any function $f\in\sBQP$ there exist polynomial-time computable
functions $\alpha$ and $\beta$, and a function $g\in\sP$, such that $f =
\alpha\circ g \circ \beta$. This differs from Karp reductions where no
postprocessing is allowed, $\alpha=\mathrm{Id}$, but it still only
requires a single call to a \sP{} oracle.
}.
In physical terms, this shows that counting
the number of ground states or determining the density of states for a
quantum Hamiltonian is not harder than either problem is for a classical
Hamiltonian.

Clearly, \sP{} is contained in \sBQP, as the latter includes classical
verifiers.  It remains to be shown that any \sBQP{} problem can be
solved by computing a \sP{} function.  We start from a verifier operator
$\Omega$, Eq.~\eqref{eq:Omega}, and wish to show that the dimension of its
accepting subspace, i.e., the subspace $\mc A$ with eigenvalues $\ge a$,
can be computed by counting satisfying inputs to some efficiently
computable boolean function. Using amplification, we can ensure
that
$\abs{\dim{\mathcal A} - \tr \Omega}\le \tfrac14$,
i.e., we need to compute $\tr \Omega$ to accuracy $\tfrac14$. This can be
done using a ``path integral'' method, which has been used previously to
show containments of quantum classes in the classical classes \PP{} and
\sP~(see e.g.\ \cite{adleman:pathintegral}). We rewrite $\tr
\Omega=\sum_\zeta f(\zeta)$ as a sum over products of transition probabilities
along a path $\zeta\equiv (i_0,\dots,i_N,j_1,\dots,j_N)$,
where
\begin{align}
\label{E:pathintegral}
f(\zeta) =& 
	    \bra{i_0}_I \braanc U_1^\dagger\proj{j_1} U_1^\dagger \cdots 
	    U_T^\dagger \ket{j_T} \times \\
	&\quad\bra{i_T}\big[{\proj 0}_1\otimes\one\big]\proj{i_T}
	     U_T \cdots  U_1  \ket{i_0}_I \ketanc \nonumber
\end{align}
(cf.~Fig.~\ref{F:circuit}). Since any such sum over an
efficiently computable $f(\zeta)$ can be determined by counting the
satisfying inputs to some boolean formula (see the Supplementary Material 
for details), it follows that $\Omega$ can be computed
using a single query to a black box solving \sP{} problems.

\emph{Summary and discussion.---}%
In this work, we considered two problems: Computing the density of states
and computing the ground state degeneracy of a local Hamiltonian of a spin
system.  In order to capture the computational difficulty of these
problems we introduced the quantum complexity class \sBQP, the counting
version of the class \QMA. We proved that this complexity class exactly
captures the difficulty of our two problems, even when restricting to
local Hamiltonians on 2D lattices of qubits or to
1D chains, since all these problems are complete problems for
the class \sBQP~\footnote{Note that computing the density of states to
multiplicative accuracy is less difficult, see
Refs.~\cite{brandao:thesis,osborne:1d-density-of-states}.}.

We have further shown that \sBQP{} is no harder than its
classical counterpart \sP. In particular this implies that 
computing the density of states is no harder for quantum Hamiltonians 
than it is for classical ones.  While this quantum-classical equivalence 
might seem surprising at the Hamiltonian level, 
it should be noted that the classes \sP{} and \PP{} quite 
often form natural ``upper bounds'' for many quantum \emph{and} 
classical problems. 

What about the problem of computing the density of states for fermionic
systems, such as many-electron systems? On the one hand, this problem will
be still in \sBQP{} and thus \sP, since any local fermionic Hamiltonian can
be mapped via the Jordan-Wigner transform to a (non-local) Hamiltonian
on a spin system, whose energy can still be estimated efficiently by a
quantum circuit~\cite{liu:n-rep-qma}. On the other hand, hardness of the
problem for \sBQP\ can be shown e.g.\ by using the \sBQP-hardness of \sLH, and
encoding each spin using one fermion in two modes, similar
to~\cite{liu:n-rep-qma}. Thus, computing the density of states for
fermionic systems is a \sBQP-complete problem as well.

\emph{Acknowledgments.---}%
We thank S.\ Aaronson, D.~Gottes\-man, D.\ Gross, Y.\ Shi, M.\ van
den Nest, J.~Watrous and S.\ Zhang for helpful discussions and comments.
Research at PI is supported by the Government of Canada through Industry
Canada and by the Province of Ontario through the Ministry of Research and
Innovation.  NS is supported by the Gordon and Betty Moore Foundation
through Caltech's Center for the Physics of Information and NSF Grant No.
\mbox{PHY-0803371}.  

\emph{Note added.---}%
After completion of this work, we learned that Shi and
Zhang~\cite{shi-zhang-private} have independently defined \sBQP\ and
shown its relation to $\sP$ using the same technique.

\onecolumngrid

\vspace*{1.0cm}

\appendix
\section*{SUPPLEMENTARY MATERIAL}

\vspace*{0.5cm}

\twocolumngrid

\subsection{Quantum Complexity Classes}

We start with the relevant definitions.  Let $x$ be a binary string.
Then, we denote by the verifier $V\equiv V_x$ a quantum circuit of length
$T=\poly(|x|)$, $U=U_T\cdots U_1$ (with local gates $U_t$) acting on
$m=\poly(|x|)$ qubits, which is generated uniformly from $x$.  
The verifier takes an $n=\poly(|x|)$ qubit quantum state $\ket{\psi}_I$ as
an input (we will express everything in terms of $n$ instead of $|x|$ in
the following), together with $m-n$ initialized ancillas,
$\ketanc\equiv\ket{0\cdots0}_A$, applies $U$, and finally measures the
first qubit in the $\{\ket{0}_1,\ket{1}_1\}$ basis to return $1$ (``proof
accepted'') or $0$ (``proof rejected''). 
The
acceptance probability for a proof $\ket\psi$ is then given by
$p_\mathrm{acc}(V(\psi)):= \bra\psi \Omega \ket\psi$, with
\begin{equation} \label{app:eq:Omega} 
\Omega=(\one_I\otimes\braanc) U^\dagger
        (\proj{1}_1\otimes\one) U(\one_I\otimes\ketanc)\, .
\end{equation}

\begin{definition}[\QMA]
Let $0\le a,b\le 1$ s.th.\ $ a-b > \frac{1}{p(n)} $ for some polynomial
$p(n)$. A language $L$ is in $\emph{\textsf{QMA}}(a,b)$ if there exists a
verifier s.th.
\begin{align*}
x\in L\ \Rightarrow & \ \lambda_{\max}(\Omega)>a\\
x\notin L\ \Rightarrow & \ \lambda_{\max}(\Omega)<b\ .
\end{align*}
\end{definition}

\begin{definition}[\sBQP]
\label{sBQP1}  
Let $0\le a,b\le 1$ s.th.\ $a-b>1/\poly(n)$, and
let $\Omega$ be a verifier map with no eigenvalues between $a$ and $b$.
Then, the class $\emph{\sBQP}(a,b)$ consists of all problems of the form
``compute the dimension of the space spanned by all eigenvectors of
$\Omega$ with eigenvalues $\ge a$''\ .
\end{definition}

An alternative definition for \sBQP{} is the following:

\begin{definition}[\sBQP, alternate definition]
\label{sBQP2}
Consider a verifier $\Omega$ with the property that there exist subspaces
$\mathcal A\oplus \mathcal R = \mathbb C^{N}$ ($N=2^n$) such that
$\bra{\psi}\Omega\ket\psi\ge a$ for all $\ket\psi\in \mathcal A$, and
$\bra{\psi}\Omega\ket\psi\le b$ for all $\ket\psi\in \mathcal R$, where
again $a-b>1/\poly(n)$.  Then \emph{\sBQP{}$(a,b)$} consists of all problems of the
form ``compute $\dim \mc A$''.
\end{definition}

Note that $\dim\,\mc A$ is well-defined: Consider two decompositions
$\mathbb C^{N} =\mathcal A\oplus \mathcal R$ and $\mathbb C^{N} =
\mathcal A'\oplus \mathcal R'$.  Without loss of generality, if we assume
\ $\dim\mathcal A > \dim \mathcal A'$, it follows $\dim \mc A + \dim \mc
R' > N$, and thus there exists a non-trivial $\ket\mu\in\mc A\cap \mc
R'$, which contradicts the definition.

\begin{theorem}
Definition \ref{sBQP1} and Definition \ref{sBQP2} are equivalent.
\end{theorem}

\begin{proof}
To show that Definition \ref{sBQP1} implies Definition \ref{sBQP2}, let
$\mathcal A$ be spanned by the eigenvectors with eigenvalues $\ge a$.  To
show the converse, we use
the minimax principle for eigenvalues~\cite{bhatia:book}, which 
states that the $k$th largest eigenvalue $\lambda_k$ of a Hermitian
operator $\Omega$ in an $N$-dimensional Hilbert space can be obtained from
either of the equivalent optimizations
\begin{align}
	\lambda_k(\Omega) 
	& = \max_{\mc M_k} \min_{\ket{x} \in \mc M_k} \bra{x} \Omega \ket{x} 
		\label{E:maxmin}\\
	& = \min_{\mc M_{N-k+1}} \max_{\ket{x} \in \mc M_{N-k+1}} \bra{x} \Omega \ket{x}
		\label{E:minmax} \,,
\end{align}
where $\mc M_k$ is a subspace of dimension $k$, and $\ket{x}$ is a unit
vector. Now notice that Def.~\ref{sBQP2} implies that
\begin{align}
	\min_{\ket{x} \in \mc A} \bra{x} \Omega \ket{x} \ge a \quad \mbox{and} \quad
	\max_{\ket{x} \in \mc R} \bra{x} \Omega \ket{x} \le b \,.
\end{align}
Next, consider the minimax theorem for $k=\dim \mc A$. From
Eq.~(\ref{E:maxmin})
we have \begin{align*} \lambda_{\dim \mc A} & =
        \max_{\mc M_{\dim \mc A}}\; \min_{\ket{x} \in \mc M_{\dim \mc A}}
        \bra{x} \Omega \ket{x} \\ & \ge \min_{\ket{x} \in \mc A} \bra{x}
        \Omega \ket{x} \\ & \ge a \,.  \end{align*}
Now consider the case that $k=\dim \mc A + 1$. From Eq.~(\ref{E:minmax}), using the fact that $N-(\dim \mc A + 1)+1 = \dim \mc R$, we have
\begin{align*}
	\lambda_{\dim \mc A+1} & = \min_{\mc M_{\dim \mc R}}\; 
            \max_{\ket{x} \in \mc M_{\dim \mc R}} \bra{x} \Omega \ket{x} \\
	& \le \max_{\ket{x} \in \mc R} \bra{x} \Omega \ket{x} \\
	& \le b \,.
\end{align*}
Thus we have
\begin{align}
	\lambda_{\dim \mc A} \ge a > b \ge \lambda_{\dim \mc A+1} \,,
\end{align}
since $a-b \ge 1/\poly(n)$. It follows that $\lambda_{\dim \mc A}$ is the smallest eigenvalue of $\Omega$ which is still larger than $a$, and therefore the span of the first $\dim \mc A$ eigenvectors of $\Omega$ is equal to the span of all eigenvectors with eigenvalue $\ge a$. The equivalence follows.
\end{proof}

The class $\sBQP(a,b)$ inherits the useful property of strong error
reduction from \QMA{}: the thresholds $(a,b)$ can be amplified to
$(1-2^{-r},2^{-r})$, $r=\poly(n)$ while keeping the size of the proof:

\begin{theorem}
        \label{thm:strongamp}
Let $\emph{\sBQP}_n(a,b)$ denote \emph{\sBQP{}} with an $n$ qubit
witness. Then $\emph{\sBQP}_n(a,b) \subset
\emph{\sBQP}_n(1-2^{-r},2^{-r})$ for every $r \in \poly(n)$.
\end{theorem}

This follows directly from the strong amplification procedure presented
in~\cite{marriott:qma-strongamp}, which describes a procedure to amplify
any verifier map $\Omega$ such that any eigenvalue above $a$ (below $b$)
is shifted above $1-2^{-r}$ (below $2^{-r}$) at an overhead polynomial in
$r$. 

Using Thm.~\ref{thm:strongamp}, we will write \sBQP\ instead of
$\sBQP(a,b)$ from now on, where $a-b>\poly(n)$, and $a$, $b$ can be
exponentially close to $1$ and $0$, respectively.

\subsection{The Complexity of the Density of States}

We now use the class \sBQP{} to characterize the complexity of the density of states
problem and the problem of counting the number of ground states of a local
Hamiltonian. We start by defining these problems, as well as the notion of local Hamiltonian, and then show that both problems are \sBQP{}-complete.

\begin{definition}[$k$-local Hamiltonian]\label{def:klocalH}
Given a set of $\poly(n)$ quantum spins each with dimension bounded by a
constant, a Hamiltonian $H$ for the system is said to be $k$-local if $H =
\sum_i H_i$ is a sum of at most $\poly(n)$ Hermitian operators $H_i$,
$\|H_i\|\le 1$,  each of which acts nontrivially on at most $k$ spins. 
\end{definition}

Note that $k$-local does not imply any geometric locality, only that each spin interacts with at most $k-1$ other spins for any given interaction term. However, we restrict ourselves to $k = O\bigl(\log(n)\bigr)$ so that each $H_i$ can be specified by an efficient classical description.

\begin{definition}[\DOS]\label{def:dos}
Let $E_2 - E_1 > 1/\poly(n)$, $\Delta=(E_2-E_1)/\poly(n)$, and 
let $H = \sum_i H_i$ be a $k$-local Hamiltonian 
such that $H$ has no eigenvalues in the intervals
$[E_1-\tfrac\Delta2,E_1+\tfrac\Delta2]$ and
$[E_2-\tfrac\Delta2,E_2+\tfrac\Delta2]$. Then, the problem \DOS{} (density
of states) is to compute the number of orthogonal eigenstates with
eigenvalues in the interval $[E_1,E_2]$. 
\end{definition}

\begin{definition}[\sLH]
Let $E_0\le E_1<E_2$, with $E_2-E_1> 1/\poly(n)$, and let
$H = \sum_i H_i$ be a $k$-local Hamiltonian s.th.\ 
$H\ge E_0$, and $H$ has no eigenvalues between $E_1$ and $E_2$.
Then, the problem $\emph{\sLH}\equiv\emph{\sLH}(E_1-E_0)$ (sharp local
Hamiltonian) is to compute the dimension of the eigenspace with
eigenvalues $\le E_1$.
\end{definition}

Note that \sLH{} depends on the ``energy splitting'' $E_1-E_0$ 
of the low-energy subspace. In particular, for $E_1-E_0=0$, $\sLH(0)$
corresponds to computing the degeneracy of the ground state subspace. As we
will see in what follows, the class $\sLH(\sigma)$ is the same for any
splitting $\exp(-\poly(n))\le \sigma \le 1/\poly(n)$.

We now show that \sLH\ and \DOS\ are both \sBQP-complete. We do so by
giving reductions from $\sLH(1/\poly(n))$\ to \DOS, from \DOS\ to \sBQP,
and from \sBQP\ to $\sLH\bigl(\exp(-\poly(n))\bigr)$; this will at the same time
prove the claimed independence of $\sLH(\sigma)$ of the 
splitting $\sigma$.

\begin{theorem}
$\sLH(1/\poly(n))$ reduces to \DOS.
\end{theorem}
\begin{proof}
If we denote the parameters of the \sLH{} problem by $\tilde{E}_0,
\tilde{E}_1, \tilde{E}_2$, then we can simply relate them to the
parameters $E_1, E_2, \Delta$ of a \DOS{} problem by $\Delta = \tilde{E}_2
- \tilde{E}_1$, $E_1 = \tilde{E}_0-\tfrac12\Delta$ and $E_2 =\tilde
E_1+\tfrac12\Delta$, and the result follows directly.
\end{proof}

\begin{theorem}
\emph{\DOS{}} is contained in \emph{\sBQP}.
\end{theorem}

\begin{proof}
We start with a $k$-local Hamiltonian $H$ as in Def.~\ref{def:dos}. Now
define a new Hamiltonian 
\[
H' := \nu(H^2 - (E_1+E_2) H + E_1E_2)\ .
\]
$H'$ is a $2k$-local Hamiltonian; here, $\nu=1/\poly(n)$ is chosen such
that each term in $H'$ is subnormalized. Any eigenvalue of $H$ in the
interval $[E_1+\tfrac\Delta2;E_2-\tfrac\Delta2]$ translates into an
eigenvalue of $H'$ which is below
\[
A:=-\nu\tfrac\Delta2(E_2-E_1-\tfrac\Delta2)\le -1/\poly(n)\ ,
\] 
whereas any eigenvalue outside $[E_1-\tfrac\Delta2;E_2+\tfrac\Delta2]$
translates into an eigenvalues of $H'$ above
\[
B:=\nu\tfrac\Delta2(E_2-E_1+\tfrac\Delta2)\ge 1/\poly(n)\ .
\]

The original \DOS\ problem now translates into counting the number of
eigenstates of $H'$ with negative energy, given a spectral gap in a 
$1/\poly(n)$ sized interval $[A;B]$ around zero.  We now use the circuit
which was used in~\cite{kitaev:qma} to prove that $\log$-local Hamiltonian
is in \QMA; it accepts any input state $\ket\psi$ with probability 
\[
p(\ket\psi) = \frac12-\frac{\bra\psi H'\ket\psi}{2m}\ ,
\]
where $m=\poly(n)$ is the number of terms in $H'$. [The idea is to 
randomly pick one term $H_i'$ in the Hamiltonian and toss a coin with
probability $(1-\bra\psi H_i'\ket\psi)/2$.] Computing the answer to the
original \DOS\ problem then translates to counting the number of states
with acceptance probability $\ge a = \tfrac12+\tfrac A{2m}$, given that
there are no eigenstates between $a$ and $b=\tfrac12+\tfrac B{2m}$, and since
$a-b=(A-B)/2m\ge1/\poly(n)$, this shows that this number can be computed in
\sBQP.
\end{proof}

\begin{theorem}
\label{thm:sLH-sBQP-hard}
\emph{\sLH$\bigl(\exp(-\poly(n))\bigr)$} is \emph{\sBQP}-hard.
\end{theorem}

\begin{proof}
To show \sBQP-hardness of \sLH, we need to start with an arbitrary \QMA{}
verifier circuit $U = U_T\ldots U_1$ and construct a Hamiltonian with as
many ground states as the circuit has accepting inputs. By amplification,
we can assume that the acceptance and rejection thresholds for the
verifier are $a = 1 - \epsilon$ and $b = \epsilon$, where we can choose
$\epsilon=O(\exp(-p(n)))$ for any polynomial $p(n)$. As before, let $\mc
A$ and $\mc R$ be the eigenspaces of $\Omega$ with eigenvalues $\geq a$
and $\leq b$, respectively. Define 
\begin{align}
\label{E:UR-def}
	U[\mc R] := \{U\ket{\psi}_I \ket{\mathbf{0}}_A : \ket{\psi}_I \in \mc R\}\,
\end{align}
and denote the projector onto this space by $\Pi_{U[\mc R]}$. Notice that
for any state $\ket{\chi} \in U[\mc R]$, due to our rejection threshold
$b=\epsilon$, we have \begin{align}
\label{E:chi}
	\bra{\chi}\bigl(\ket{1}\bra{1}_1\otimes\one\bigr)\ket{\chi} \le \epsilon\,.
\end{align}

We now follow Kitaev's original construction to encode a \QMA{} verifier
circuit into a Hamiltonian which has a ``proof history'' as its
ground state for any proof $\ket{\phi}_I \in \mc
A$~\cite{kitaev:book,kitaev:qma}.  That is, the ground
states of the Hamiltonian are given by
\begin{equation}
\ket{\Phi} = \sum_{t=0}^T U_t\ldots U_1\ket{\phi}_I\ket{\mathbf{0}}_A\ket{t}_T
\end{equation}
where the third register is used as a ``clock''. The Hamiltonian has the form

\begin{equation}
        \label{eq:kitaev-ham}
H = H_{\mathrm{init}} + \sum_{t=1}^T H_{\mathrm{evol}}(t) + H_{\mathrm{final}}
\end{equation}
where
\begin{itemize}
\item $H_{\mathrm{init}} = \one_I \otimes (\one - \ket{\mathbf{0}}\bra{\mathbf{0}}_A)\otimes\ket{0}\bra{0}_T$
checks that the ancilla is property initialized, penalizing states without properly initialized ancillas;
\item 
	$H_{\mathrm{evol}}(t) =  -\tfrac{1}{2} U_t \otimes
        \ket{t}\bra{t-1}_T - \tfrac{1}{2}U_t ^{\dagger} \otimes
        \ket{t-1}\bra{t}_T$\\[1ex]
    $
	\quad\qquad\qquad + \tfrac{1}{2}\openone\otimes\ket{t}\bra{t}_T +\tfrac{1}{2}
        \openone\otimes\ket{t-1}\bra{t-1}_T$\\[1ex]
checks that the propagation from time $t-1$ to $t$
is correct, penalizing states with erroneous propagation;
\item $H_{\mathrm{final}} = \Pi_{U[\mc R]}\otimes \ket{T}\bra{T}_T$ causes
        each state $\ket{\phi}$ built from an input $\ket{\psi}_{I} \in
        \mc R$ (but which nonetheless has a correctly initialized ancilla)
        to receive an energy penalty.  
\end{itemize}

As we show in Lemma~\ref{L:blk-diag}, the total Hamiltonian $H$ has a
spectral gap $1/\poly(n)$ above the $\dim\mc A$--dimensional ground state
subspace.  However,
$H_\mathrm{final}$ is not a local Hamiltonian, but as we argue in the
following, it can be replaced by the usual term
$H_{\mathrm{final}}^{\mathrm{std}} =
\ket{0}\bra{0}_1\otimes\one\otimes\ket{T}\bra{T}_T$ while keeping the
ground space dimension (up to small splitting in energies) and the
$1/\poly(n)$ spectral gap.  As we prove in 
Lemma~\ref{L:spec-bound},
\begin{equation} \label{E:spectral-ineq} H_{\mathrm{final}}^{\mathrm{std}}
        \geq H_{\mathrm{final}} - \sqrt{\epsilon}\one \,.  \end{equation}
Thus, replacing $H_\mathrm{final}$ by
$H_\mathrm{final}^\mathrm{std}$ will decrease the energy of any excited
state by at most $\sqrt{\epsilon}=O(\exp(-p(n)/2))$. (The energy of the
ground states is already minimal and cannot decrease.)
On the other hand, the energy of any proper proof history for $H$ cannot
increase by more than
\[
\bra{\chi} H_{\mathrm{final}}^{\mathrm{std}} \ket{\chi}
=
\bra{\chi}\ket{0}\bra{0}_1\otimes\one\ket\chi\le
\epsilon=O(\exp(-p(n))\ ,
\]
 due to our choice of acceptance threshold, i.e., the low energy subspace
has dimension $\dim \mc A$.  We thus obtain a Hamiltonian with a $\dim \mc
A$ dimensional ground state subspace with energy splitting $\leq
\epsilon=\exp(-p(n))$ and a $1/\poly(n)$ spectral gap above.
\end{proof}

The following two lemmas are used in the preceding proof of
Theorem~\ref{thm:sLH-sBQP-hard}.
\begin{lemma}
\label{L:blk-diag}
$H$ has a spectral gap of size $1/\poly(n)$.
\end{lemma}

\begin{proof}
Our proof follows closely the discussion in Ref.~\cite{kitaev:book},
cf.~also~\cite{schuch:mps-gap-np}. We
can block diagonalize $H$ by the (conjugate) action of the following
unitary operator, 
\begin{align} W = \sum_{j=0}^T
        U_j \cdots U_1 \otimes \ket{j}\bra{j}_T \,, \end{align}
which maps $H \to H' = W^\dag H W$. As this has no effect on the spectrum,
we can work with the simpler $H'$ from now on. 

Let us explicitly write the effect of conjugation by $W$ on the terms of
$H$. The first term is unaffected, $H'_{\mathrm{init}}=H_{\mathrm{init}}$.
The final term becomes $H'_{\mathrm{final}} = \Pi_{\mc R}\otimes
\ket{\mathbf 0}\bra{\mathbf 0}_A \otimes \ket{T}\bra{T}_T$, where
$\Pi_{\mc R}$ is simply the projector onto the space $\mc R$, and the
ancillas are in the correct initial state.

We can now conjugate each of the terms in $H_{\mathrm{evol}}(t)$ separately. For example, the first term gives
\begin{align}
	W^\dag (U_t \otimes \ket{t}\bra{t-1}) W = \one \otimes \ket{t}\bra{t-1} \,.
\end{align}
The other terms are exactly analogous, and we find that
\begin{align*}
	H'_{\mathrm{evol}}(t) &= \openone \otimes
        \tfrac{1}{2}\big[\ket{t-1}\bra{t-1} + \ket{t}\bra{t}\\ 
        &\qquad\qquad-\ket{t}\bra{t-1}-\ket{t-1}\bra{t}\big]\,.
\end{align*}
The total evolution Hamiltonian is then block-diagonal a matrix which looks
like a hopping Hamiltonian in the clock register, \begin{align} \sum_t
        H'_{\mathrm{evol}}(t) = \one \otimes E \,, \end{align}
where the $(T+1)$-by-$(T+1)$--dimensional tri-diagonal matrix $E$ is given by
\begin{align}
	E = 
	\left(
	\begin{array}{rrrrr}
		\frac{1}{2} & -\frac{1}{2} & & &  \\ \noalign{\medskip}
		- \frac{1}{2} & 1 & -\frac{1}{2}  & & \\ \noalign{\medskip}
		 & - \frac{1}{2} & \ddots &  & \\ \noalign{\medskip}
		 &  &  &  1 & -\frac{1}{2}\\ \noalign{\medskip}
		 &  &  & -\frac{1}{2} & \frac{1}{2}
	\end{array}
	\right) \,.
\end{align}

To discuss the spectrum of $H'$,  let
\begin{align*}
\mc S_1&=\mc A\otimes \ket{0}\bra{0}_A\otimes \mathbb C^{T+1}\;,\\
\mc S_2&=\mc R\otimes \ket{0}\bra{0}_A\otimes \mathbb C^{T+1}\;,
\end{align*}
and $\mc S_3$ the orthogonal complement of $\mc S_1\oplus\mc S_2$; i.e.,
$\mc S_1$ corresponds to evolutions starting from good proofs, $\mc S_2$
to those starting from wrong proofs, and $\mc S_3$ to evolutions with
wrongly initialized ancillas.  Note that
$H'=H|_{\mc S_1}\oplus H|_{\mc S_2}\oplus H|_{\mc S_3}$, thus, we can
analyze the spectrum for $H|_{\mc S_p}$ separately.  Note that the
restriction to $\mc S_p$ does not affect $H'_\mathrm{evol}(t)$. Since we
expect the ground state subspace to occur on $\mc S_1$, we need to compute
ground state energy and gap of $H'|_{\mc S_1}$, as well as lower bound the
ground state energies of $H'|_{\mc S_2}$ and $H'_{\mc S_3}$.

First, $H'_\mathrm{init}|_{\mc S}=H'_\mathrm{final}|_{\mc S}=0$, i.e., the
spectrum of $H'|_{\mc S}$ equals the spectrum of $E$, which can be
straightforwardly determined to be $1-\cos\omega_n$, with
$\omega_n=n\pi/(T+1)$, and eigenfunctions $(\cos \tfrac12\omega_n,
\cos\tfrac32\omega_n,\dots)$; the ground state degeneracy is indeed $\dim
\mc A$ as desired.

On the other hand, $H'_\mathrm{final}|_{\mc S_2}=1$, and
$H'_\mathrm{init}|_{\mc S_3}\ge1$, i.e., in both cases the ground state
energy of $H'|_{\mc S_p}$ is lower bounded by the ground state energy of 
\begin{align*}
	E' = 
	\left(
	\begin{array}{rrrrr}
		\frac{1}{2} & -\frac{1}{2} & & &  \\ \noalign{\medskip}
		- \frac{1}{2} & 1 & -\frac{1}{2}  & & \\ \noalign{\medskip}
		 & - \frac{1}{2} & \ddots &  & \\ \noalign{\medskip}
		 &  &  &  1 & -\frac{1}{2}\\ \noalign{\medskip}
		 &  &  & -\frac{1}{2} & \frac{3}{2}
	\end{array}
	\right) \,,
\end{align*}
which has eigenvalues $1-\cos\vartheta_n$, with
$\vartheta_n=(n+\tfrac12)\pi/(T+\tfrac32)$, and eigenfunctions $(\cos
\tfrac12\vartheta_n, \cos\tfrac32\vartheta_n,\dots)$. 

It follows that $H'$ (and thus $H$) has a ground state energy of
$1-\cos\omega_0=0$, and a gap
$1-\cos\tfrac{\pi}{2T+3}=O(1/T^2)=O(1/\poly(n))$ above.
\end{proof}

It remains to prove Eq.~(\ref{E:spectral-ineq}), which follows from the
following Lemma by choosing
$P = \ket{0}\bra{0}_1\otimes\one$, $Q = \Pi_{U[\mc R]}$, and using
Eq.~(\ref{E:chi}).
\begin{lemma}
\label{L:spec-bound}
Let $P$ and $Q$ be projectors such that $\lVert Q (\one -P) Q
\rVert_\infty \le \epsilon$. Then 
\begin{align}
	P - Q \ge - \sqrt{\epsilon}\one \,.
\end{align}
\end{lemma}
\begin{proof}
We begin by recalling the result due to Jordan~\cite{Jordan1875} (see Ref.~\cite{Halmos1969} for a more modern treatment) for the simultaneous canonical form of two projectors. In the subspace where $P$ and $Q$ commute, both operators are diagonal in a common basis and the spectrum is either $(0,0), (0,1), (1,0)$, or $(1,1)$, and direct sums of those terms. In the subspace where they don't commute, the problem decomposes into a direct sum of two-by-two blocks given by
\begin{align}
	P_j = \begin{pmatrix}
	1&0\\
	0&0
\end{pmatrix} \quad , \quad Q_j = \begin{pmatrix}
	c^2&cs\\
	cs&s^2
\end{pmatrix}\,,
\end{align}
where $s = \sin(\theta_j)$ for some angle $\theta_j$, $c^2+s^2=1$, and the subscript $j$ just labels a generic block. 

In fact, this two-by-two block form is completely general if we allow embedding our projectors into a larger space while preserving their rank. The rank-preserving condition guarantees that our bound is unchanged, since we are only appending blocks of zeros, and so we will consider this two-by-two form without loss of generality.

The constraint that $\lVert Q (\one -P) Q \rVert_\infty \le \epsilon$ implies
constraints on the values that $\sin(\theta_j)$ can take. In
particular, we can directly compute this operator norm in each block
separately, and we find that for all $j$ \begin{align} \lVert Q_j (\one
        -P_j) Q_j \rVert_\infty = \sin^2(\theta_j)\le \epsilon \,.  \end{align}
We can also directly compute in each block that
\begin{align}
	P_j - Q_j = \begin{pmatrix}
	1- c^2&1- cs\\
	-cs&-s^2
\end{pmatrix}\,.
\end{align}
The spectrum of this operator is easily computed to be $\pm \lvert \sin(\theta_j)\rvert$. Thus, the least eigenvalue of $P-Q$ is bounded from below by $-\sqrt \epsilon$. 
\end{proof}

\subsection{Quantum vs. Classical Counting Complexity}

We finally want to relate \sBQP\ to the classical counting class \sP. It
is clear that any \sP\ problem can be solved in \sBQP\ by using a
classical circuit. We will now show that conversely, \sBQP\ can be reduced
to \sP\ under \emph{weakly parsimonious reductions}: That is, for any
function $f\in\sBQP$ there exist polynomial-time computable functions
$\alpha$ and $\beta$, and a function $g\in\sP$, such that $f = \alpha\circ
g \circ \beta$. This differs from Karp reductions where no postprocessing
is allowed, $\alpha=\mathrm{Id}$, but still only requires a single call to
a \sP{} oracle, in contrast to Turing reductions. 

\begin{theorem}
There exists a weakly parsimonious reduction from \emph{\sBQP{}} to
\emph{\sP}.
\end{theorem}
\begin{proof}
We start from a verifier operator $\Omega$ for a \sBQP\ problem.
First, we use strong error reduction to let $a = 1 - 2^{-(n-2)}$ and $b = 2^{-(n+2)}$. It
follows that
\begin{equation}
\lvert \dim \mc A - \tr \Omega \rvert \leq 2^n 2^{-(n+2)} = \frac{1}{4}
\end{equation}
and thus we need to compute $\tr \Omega$ to accuracy $1/4$. This can be
done using the ``path integral'' method previously used to show
containments of quantum classes inside \PP\ and
\sP~\cite{adleman:pathintegral}. We rewrite $\tr
\Omega  = \sum_\zeta f(\zeta)$, where the sum is over products of
transition probabilities along a path, which we label 
\begin{align}
	\zeta \equiv (i_0,\ldots,i_N,j_1,\ldots,j_N)\,,
\end{align}
so that
\begin{align}
\label{app:eq:pathintegral}
	f(\zeta) =& 
	    \bra{i_0}_I \braanc U_1^\dagger\proj{j_1} U_1^\dagger \cdots 
	    U_T^\dagger \ket{j_T} \times \\
	&\quad\bra{i_T}\big[{\proj 0}_1\otimes\one\big]\proj{i_T}
	     U_T \cdots  U_1  \ket{i_0}_I \ketanc\, .
             \nonumber
\end{align}
(cf. Fig.~1 in the main manuscript for an illustration).

Since any quantum circuit can be recast in terms of real gates at the cost
of doubling the number of
qubits~\cite{bernstein:quantum-complexity-theory}, we can  simplify the
proof by assuming $f(\zeta)$ to be real.  To achieve the desired accuracy
it is sufficient to approximate $f$ up to $|\zeta|+2$
digits, where $|\zeta| = \poly(n)$ is the number of bits in $\zeta$. Now define
\begin{equation}
g(\zeta):=\mathrm{round}\big[2^{\abs{\zeta}+2}(f(\zeta)+1)\big]
\end{equation}
and note that $g(\zeta)$ is a positive and integer-valued function satisfying
\begin{equation}
        \label{eq:one-more-approx}
	\Bigl\lvert\bigl[2^{-\abs{\zeta}-2} \sum_\zeta g(\zeta) - 1 \bigr] 
	- \sum_\zeta f(\zeta)\Bigr\rvert \le \tfrac14 \,.
\end{equation}
Finally, by defining a boolean indicator function,
\begin{align*}
	h(\zeta,\xi) = 
		\begin{cases} 1 & \mbox{if } 0\le\xi < g(\zeta)\\ 
                        0 & \mbox{otherwise}
		\end{cases}
\end{align*}
we can write $g(\zeta) = \sum_{\xi\ge0} h(\zeta,\xi)$, and thus,
\[
\sum_\zeta g(\zeta) = \sum_{\xi,\zeta} h(\zeta,\xi)\ .
\]
This shows that $\tr \Omega$
can be approximated to accuracy $\tfrac14$, and thus $\dim \mc A$ can be
determined by counting the number of satisfying assignments of a single
boolean function $h(\zeta,\xi)$ that can be efficiently constructed from
$\Omega$, i.e., by a single query to a black box solving \sP{}
problems, together with the efficient postprocessing described by
Eq.~(\ref{eq:one-more-approx}). 
\end{proof}


\end{document}